\documentclass[12pt]{amsart}

\usepackage[colorinlistoftodos,bordercolor=orange,backgroundcolor=orange!20,linecolor=orange,textsize=scriptsize]{todonotes}

\usepackage[english]{babel}
\usepackage{amsmath, amsthm, amssymb}
\usepackage{enumitem}
\usepackage{tikz}
\usepackage{fullpage}
\usepackage{hyperref}

\newtheorem{proposition}{Proposition}

\newtheorem{lemma}{Lemma}
\newtheorem{theorem}{Theorem}

\theoremstyle{definition}
\newtheorem{remark}{Remark}

\newcommand{\Drond}{\mathcal{D}}

\newcommand{\sYj}{s_{Y_j}}
\newcommand{\sXj}{s_{X_j}}
\newcommand{\SUj}{S_{U_j}}

\begin{document}

\title{Perfect graphs with polynomially computable kernels}

\author{Ad\`ele Pass-Lanneau}
\address{A. Pass-Lanneau and F. Meunier -- Universit\'e Paris-Est, CERMICS (ENPC), 6-8 avenue Blaise Pascal, 77455 Marne-la-Vall\'ee, France}
\email{adele.pass-lanneau@polytechnique.edu}
\email{frederic.meunier@enpc.fr}

\author{Ayumi Igarashi}
\address{A. Igarashi -- University of Oxford, Department of Computer Science, United Kingdom}
\email{ayumi.igarashi@cs.ox.ac.uk}

\author{Fr\'ed\'eric Meunier}


\begin{abstract}
In a directed graph, a kernel is a subset of vertices that is both stable and absorbing. Not all digraphs have a kernel, but a theorem due to Boros and Gurvich guarantees the existence of a kernel in every clique-acyclic orientation of a perfect graph.
However, an open question is the complexity status of the computation of a kernel in such a digraph. Our main contribution is to prove new polynomiality results for subfamilies of perfect graphs, among which are claw-free perfect graphs and chordal graphs.
Our results are based on the design of kernel computation methods with respect to two graph operations: clique-cutset decomposition and augmentation of flat edges.
We also prove that deciding the existence of a kernel -- and computing it if it exists -- can be done in polynomial time in any orientation of a chordal or a circular-arc graph, even not clique-acyclic.
\end{abstract}


\maketitle

\section{Introduction}
A subset $S$ of vertices of a digraph $D=(V,A)$ is \emph{stable} if it contains no pair of adjacent vertices,
 and \emph{absorbing} if for any vertex $u$ not in $S$, there is a vertex $v \in S$ such that the arc $(u,v)$ exists in $D$. A \emph{kernel} is a subset of vertices that is both stable and absorbing. Kernels have been introduced in 1944 by Von Neumann and Morgenstern \cite{VonNeumannMorgenstern} as a tool for
the study of positional or Nim-type games. Since then, other applications in game theory have
been found \cite{BorosGurvichAppli}. Kernels also play a role in graph theory: they are for instance
at the heart of Galvin's proof of Dinitz's conjecture on list coloring \cite{Galvin}. Not all digraphs admit a kernel: for example a directed cycle of length three has no kernel. It has been shown by Chv\'atal \cite{ChvatalNPC} that deciding if a digraph has a kernel is NP-complete. It is even NP-complete in planar graphs with small degree \cite{Fraenkel}.

An \emph{orientation} of an undirected graph $G$ is a digraph $D$ obtained by orienting every edge of $G$ in only one direction. A {\em super-orientation} of $G$ is a digraph $D$ obtained by orienting each edge of $G$ either in one direction or in both directions. An orientation or a super-orientation of a graph is \emph{clique-acyclic} if no clique has a directed cycle consisting of arcs oriented in only one direction, or equivalently, if every clique has a \emph{sink}, i.e., a vertex absorbing all other vertices of the clique. Note that in case of super-orientations, a clique may have several sinks.

One of the main results in the field is a theorem proved in 1996 by Boros and Gurvich~\cite{BorosGurvich} -- originally conjectured by Berge and Duchet in 1980 -- which states that every clique-acyclic super-orientation of a perfect graph has a kernel. (We remind the reader that a graph is {\em perfect} if the clique number is equal to the chromatic number for this graph and for all its induced subgraphs.) The original proof was quite involved and relied on advanced notions of game theory. A simpler and much more direct proof based on Scarf's lemma was later proposed by Aharoni and Holzman~\cite{AhHo} and further simplified by Kir\'aly and Pap~\cite{KiralyPap} using Sperner's lemma. However, none of these proofs provides any efficient way of computing a kernel. The question of the complexity of kernel computation in a clique-acyclic super-orientation of a perfect graph has been identified as a challenging open problem~\cite{Egres}.

Polynomial-time algorithms to compute kernels are known in some special cases: for instance, digraphs without directed cycle (``select the sinks and recurse''), digraphs without odd directed cycle (Richardson's algorithm~\cite{Richardson}), clique-acyclic orientations of permutation graphs~\cite{Abbas}, and clique-acyclic super-orientations of line-graphs of bipartite graphs. As noted by Maffray~\cite{maffray1992kernels}, in this latter case  kernels coincide with stable matchings in bipartite graphs, computable through the celebrated algorithm of Gale and Shapley~\cite{GaleShapley}. Permutation graphs and line-graphs of bipartite graphs are perfect. With the result by Durand de Gevigney et al.~\cite{Meunier} about so-called DE graphs (see Section~\ref{sec:inter}), this is almost all that is known about kernel computation in clique-acyclic super-orientations of perfect graphs. 

We contribute to the literature with the following two theorems, which address the problem for two special classes of perfect graphs: claw-free perfect graphs and chordal graphs. A graph is {\em claw-free} if no three vertices with a common neighbor form a stable set. 
A graph is {\em chordal} if there is no induced cycle of length $4$ or more.

\begin{theorem}\label{clawFree}
Kernel computation is polynomial in clique-acyclic orientations of claw-free perfect graphs.
\end{theorem}

\begin{theorem}\label{ChordalCASuper}
Kernel computation is polynomial in clique-acyclic super-orien\-ta\-tions of chordal graphs.
\end{theorem}

In both cases, we show that there exists an algorithm that, when given a graph of the mentioned classes, returns a kernel in polynomial time. While our result for chordal graphs is about super-orientations, we had to restrict ourselves to the case of orientations for claw-free graphs. The complexity status of kernel computation in super-orientations of claw-free graphs is left as an open problem. Note also that in the case of Theorem~\ref{ChordalCASuper} if only clique-acyclic orientations are considered, the computation of a kernel is easy: indeed a clique-acyclic orientation of a chordal graph is a digraph without directed cycle. 

We end this introduction with a few insights in the light of complexity theory.
For each of our two theorems, deciding whether the instance belongs to the expected class can be done in polynomial time: for orientations, clique-acyclicity is equivalent to the absence of directed cycle of length three; for super-orientations of chordal graphs, clique-acyclicity can be checked by enumerating maximal cliques, since there is a linear number of such cliques. In general, deciding whether a super-orientation of a perfect graph is clique-acyclic is a coNP-complete problem~\cite[Corollary 11]{AnHo14}. Therefore, a way to get a problem with polynomially recognizable instances is to restrict to orientations of perfect graphs. However, the problem of computing a kernel in an orientation of a perfect graph is not known to belong to any of the classical subclasses of TFNP, such as PPAD or PLS, except if the graphs have bounded cliques, as done by Kintali et al.~\cite{kintali2009reducibility} for a related question (complexity of computing a ``fractional kernel''; we also refer to that paper for definitions and discussions on these complexity classes).


\begin{remark}
A preliminary version of this paper with Theorem~\ref{ChordalCASuper} and Theorem~\ref{thm:DE} (in Section~\ref{sec:inter}) has been presented at CTW 2017.
\end{remark}

\section{Graph operations and kernels}
\label{sec:Tools}

This section gathers preliminary results that will be useful for the proofs of Theorems~\ref{clawFree} and~\ref{ChordalCASuper}. The following notation will be used throughout the paper. In a graph $G$, the set of all neighbors of a vertex $v$ is denoted by $N_G(v)$. In a digraph $D$, the set of all out-neighbors (resp. in-neighbors) of a vertex $v$ is denoted by $N_D^+(v)$ (resp. $N_D^-(v)$).

\subsection{Clique-cutset}
A \emph{clique-cutset} of a digraph $D = (V, A)$ is a subset $C \subseteq V$ such that $C$ induces a clique in $D$ and $D[V \backslash C]$ is disconnected.
For each connected component $B$ of $D[V \backslash C]$, the directed graph induced by $B \cup C$ is a \emph{piece} of $D$ with respect to $C$. 
The main result of this section -- Proposition~\ref{prop:atomes} -- claims, roughly speaking, that if we are able to compute efficiently kernels in pieces, we are then able to combine them efficiently to get a kernel in the whole graph. It relies on the following lemma.

\begin{lemma}\label{lem:cliqueCutset}
Let $C$ be a clique-cutset of a digraph $D=(V,A)$ and $B,B'$ a bipartition of $V\setminus C$ such that $D[B\cup C]$ is a piece of $D$ with respect to $C$. Suppose that there exist subsets $K_i$ of vertices such that $K_i$ is a kernel of $D\left[B\cup \left(C\setminus\bigcup_{j=1}^{i-1}K_j\right)\right]$ for every $i\in\{1,\ldots,|C|+1\}$
and such that $D\left[B'\cup \left(C\cap\bigcup_{j=1}^{|C|}K_j\right)\right]$ has a kernel $K$. Then there exists $i\in\{1,\ldots,|C|+1\}$ such that $K\cup K_i$ is a kernel of $D$.
\end{lemma}

\begin{proof}
For notational simplicity, let us define $X_i:=C\cap\bigcup_{j=1}^{i-1}K_j$, so that $K_i$ is a kernel of $D[B \cup (C \backslash X_i)]$ for every $i\in\{1,\ldots,|C|+1\}$.

We first prove that there is an index $k\in\{1,\ldots,|C|+1\}$ such that $C\cap K_k=\varnothing$. The sequence $(X_i)_{i=1,\ldots,|C|+2}$ of subsets of $C$ is non-decreasing (for inclusion). There is thus an index $k\in\{1,\ldots,|C|+1\}$ such that $X_k=X_{k+1}$. Since $K_k$ is a kernel of $D[B\cup (C\setminus X_k)]$, we have $C\cap K_k\subseteq C\setminus X_k$. The equality $X_k=X_{k+1}$ implies that $C \cap K_k\subseteq X_k$. We have thus  $C\cap K_k=\varnothing$, as claimed.

Assume that $D[B'\cup X_{|C|+1}]$ has a kernel $K$. Suppose first that $C$ and $K$ have an empty intersection. The set $C$ being a clique-cutset and $C\cap K_k$ being empty, the set $K\cup K_k$ is stable. The vertices in $B\cup (C\setminus X_k)$ are absorbed by $K_k$ and those of $B'\cup X_k$ are absorbed by $K$ because $X_k \subseteq X_{|C|+1}$. Hence $K\cup K_k$ is a kernel of $D$.

Suppose then that $C$ and $K$ have a non-empty intersection. Denote by $v$ a vertex in $C\cap K$. By definition of $K$, we have $C\cap K\subseteq X_{|C|+1}$ and $v \in \bigcup_{j=1}^{|C|} K_j$. There is thus an index $\ell\in\{1,\ldots,|C|\}$ such that $v\in K_{\ell}$. Since $C$ is a clique-cutset, the set $K\cup K_{\ell}$ is stable in $D$.  The vertices in $B\cup (C\setminus X_{\ell})$ are absorbed by $K_{\ell}$ and those of $B'\cup X_{\ell}$ by $K$, we get that $K\cup K_{\ell}$ is a kernel of $D$.
\end{proof}

\begin{proposition}
\label{prop:atomes}
Consider two classes $\Drond$ and $\Drond_0$ of clique-acyclic super-orientations of perfect graphs, closed for taking induced subdigraphs, such that every digraph of $\Drond$ without a clique-cutset belongs to $\Drond_0$.
Suppose there exists an algorithm that, when given a digraph of $\Drond_0$, computes a kernel of this digraph in polynomial time. 

Then there exists an algorithm that, when given a digraph of $\Drond$, computes a kernel of this digraph in polynomial time.
\end{proposition}

\begin{proof}
Let $\mathbf{A}_0$ be the polynomial-time algorithm whose existence is assumed.
Let $D$ be any digraph in $\Drond$. Let us describe the general algorithm $\mathbf{A}$ to apply to $D$ to compute a kernel.

If $D$ has no clique-cutset, apply $\mathbf{A}_0$; otherwise compute a clique-cutset $C$ of $D$ and a piece $D[B \cup C]$ of $D$ with respect to $C$ such that $D[B \cup C]$ contains no clique-cutset. If $D$ has at least one clique-cutset, then such a piece exists, and as shown by Tarjan~\cite{Tarjan}, it can be detected in polynomial time.

Compute $K_1,\ldots,K_{|C|+1}$ as in the statement of Lemma~\ref{lem:cliqueCutset}. By assumption, they are computable in polynomial time using the algorithm $\mathbf{A}_0$ since $D[B \cup C]$ is in $\Drond_0$. Call recursively algorithm $\mathbf{A}$ on $D\left[B'\cup \left(C\cap\bigcup_{j=1}^{|C|}K_j\right)\right]$, which is an element of $\Drond$. Given its result $K$, a kernel $K  \cup K_i$ of $D$ is found according to Lemma~\ref{lem:cliqueCutset}.

Let us prove that $\mathbf{A}$ has polynomial complexity. Denote by $f(n)$ it complexity function for an input graph with $n$ vertices. In addition to the computation of $K$ (whose complexity is upper-bounded by $f(n-1)$), the algorithm has to perform the following tasks: 
 find the clique $C$ and the set $B$, perform $\mathbf{A}_0$, and find $i$ such that $K\cup K_i$ is the kernel of $D$. We can find positive values $\lambda$ and $\alpha$ with $\lambda\geq f(1)$ such that the complexity of each of these tasks is upper-bounded by $\lambda n^\alpha$.
 We have for $n\geq 2$
$$ f(n) \leq 2\lambda n^\alpha + (|C|+1) \lambda n^\alpha+ f(n-1)\leq \lambda n^{\alpha+2}+ f(n-1).$$ 
A direct induction shows that $f(n)\leq\lambda n^{\alpha+3}$.
\end{proof}

\begin{remark}
A special version of Lemma~\ref{lem:cliqueCutset} is a neat theorem by Jacob~\cite{Jacob}, used by Maffray for proving his result about kernels in $i$-triangulated graphs \cite{MaffrayITriang}. A digraph $D$ is \emph{kernel-perfect} if every subdigraph of $D$ has a kernel. Jacob's theorem states that if every piece of digraph $D$ with respect to a given clique-cutset is kernel-perfect, then $D$ has a kernel.

Jacob's theorem provides an existence result, but it does not give any clue regarding complexity: his proof can be adapted into an algorithm to compute a kernel of a digraph, but we were not able to make it polynomial, even on simple instances such as clique-acyclic orientations of interval graphs. 
Lemma~\ref{lem:cliqueCutset} relies on a different method to combine kernels of the pieces and its statement makes it amenable to algorithmic approaches. We may also note that our proof is much shorter than Jacob's original proof.
\end{remark}

\subsection{Augmentation of an edge}

An edge $xy$ in an undirected graph is \emph{flat} if it is not contained in any triangle. A flat edge can be subject to an augmentation, which is an operation introduced by Maffray and Reed for their characterization of claw-free perfect graphs~\cite{MaffrayReed}. A formal definition is given as follows.

Let $xy$ be a flat edge in a graph $H=(W,F)$ and let $B=(X,Y;E_{XY})$ be a cobipartite graph, with $E_{XY}$ being non-empty. We remind the reader that a cobipartite graph is the complement of a bipartite graph. Its vertex-set is partitioned into the two cliques $X$ and $Y$, and $E_{XY}$ is the set of edges between $X$ and $Y$. We assume that $B$ is disjoint from $H$.

The {\em augmentation of the edge $xy$ in $H$} consists in building a new graph from $H$ and $B$ by removing from $H$ the edge $xy$ and the vertices $x$ and $y$, and by adding all possible edges between $X$ and $N_H(x)\setminus\{y\}$ and between $Y$ and $N_H(y)\setminus\{x\}$. 

Consider a clique-acyclic orientation $D=(V,A)$ of a graph $G$ obtained by an augmentation of a flat edge $xy$ in a graph $H$. Denote by $s_X$ and $s_Y$ the sinks of $D[X]$ and $D[Y]$ respectively. Assume w.l.o.g. that $(s_X,s_Y)\notin A$, i.e., if there is an arc between $s_X$ and $s_Y$ it is the arc $(s_Y,s_X)$. Let $U=Y\setminus N_G(s_X)$ and let $s_U$ be the sink of $D[U]$ if $U \not= \varnothing$. We define 
$$S_U=\left\{\begin{array}{ll} \varnothing & \mbox{if $U=\varnothing$} \\ \{s_U\} & \mbox{otherwise}\end{array}\right.\qquad\mbox{and}\qquad Z=V\setminus(X\cup Y)\cup\{s_X,s_Y\}\cup S_U.$$

The following two results play a crucial role in our proof of Theorem~\ref{clawFree}.

\begin{lemma}\label{obs:structureGZ}
The graph $G[Z]$ is isomorphic to $H$, or to $H\setminus\{xy\}$, or to $H$ plus an additional vertex with neighborhood $N_H(y)\cup\{y\}\setminus\{x\}$.
\end{lemma}

\begin{lemma}\label{lem:augmentation}
Any kernel of $D[Z]$ is a kernel of $D$.
\end{lemma}
\begin{proof}[Proof of Lemma~\ref{obs:structureGZ}]
In all cases, vertex $s_X$ (resp. $s_Y$) of $G[Z]$ is identified with vertex $x$ (resp. $y$) in $H$.
If $(s_Y, s_X) \notin A$, then $s_X$ and $s_Y$ are not neighbors and $s_U =s_Y$. Then $G[Z]$ is isomorphic to $H \backslash \{xy\}$.
If $(s_Y,s_X) \in A$ and $U = \varnothing$, then $S_U = \varnothing$ and $G[Z]$ is isomorphic to $H$.
If $(s_Y,s_X) \in A$ and $U \not= \varnothing$, then $G[Z]$ is isomorphic to $H$ plus the additional vertex $s_U$. Furthermore, $s_U \in U$ implies that the neighborhood of this additional vertex is exactly $N_H(y) \cup \{y\} \backslash \{x\}$.
\end{proof}

\begin{proof}[Proof of Lemma~\ref{lem:augmentation}]
Let $K$ be a kernel in $D[Z]$. Since $D[Z]$ is an induced subdigraph of $D$, it is clear that $K$ remains stable in $D$. Furthermore, it is sufficient to prove that $K$ absorbs $V \setminus Z$ to get that $K$ is a kernel of $D$.

The conclusion will be almost straightforward with the following claim: {\em Let $v\in V$. If there is a vertex $z\in N_D^+(v)\cap Z$ such that $N_{D[Z]}^+(z)\subseteq N_G(v)$, then $v$ is absorbed by $K$.} This claim is true because if $z$ is not in $K$, then $z$ is absorbed by a vertex $k\in N_{D[Z]}^+(z) \cap K \subseteq N_G(v) \cap K$, and clique-acyclicity implies that $k$ absorbs $v$ as well.

Consider now a vertex $v\in V\setminus Z$. To finish the proof, we have to show that there always exists a vertex in $K$ absorbing $v$.

Assume first $v \in X \setminus \{s_X\}$. By definition of $s_X$, we have $s_X\in N_D^+(v)\cap Z$. By assumption $(s_X, s_Y) \notin A$ and if the vertex $s_U$ exists, it is not a neighbor of $s_X$. Hence $N_{D[Z]}^+(s_X)\subseteq N_H(x) \setminus \{y\}$. By definition of augmentation, we have thus $N_{D[Z]}^+(s_X)\subseteq N_G(v)$ and the claim allows to conclude with $z=s_X$.

Assume then $v \in U \setminus \{s_U\}$. We have $s_U\in N_D^+(v)\cap Z$. The vertex $s_U$ is not a neighbor of $s_X$ by definition of $U$. Hence $N_{D[Z]}^+(s_U)\subseteq N_H(y) \backslash \{x\} \cup \{s_Y\}$. By definition of the augmentation, we have thus $N_{D[Z]}^+(s_U)\subseteq N_G(v)$ and the claim allows to conclude with $z=s_U$.

Assume finally $v \in Y \setminus\left(U\cup\{s_Y\}\right)$. We have $s_Y\in N_D^+(v)\cap Z$ and $N_{D[Z]}^+(s_Y)\subseteq N_H(y) \backslash\{x\} \cup \{s_X\}$. By definition of the augmentation, we have thus $N_{D[Z]}^+(s_Y)\subseteq N_G(v)$ and the claim allows to conclude with $z=s_Y$.
\end{proof}

\section{Proof of Theorem~\ref{clawFree}}
\label{sec:clawFree}

An \emph{augmentation of a graph} is defined by taking a matching of flat edges in the graph and applying sequentially augmentations of the edges in the matching, as presented in Section~\ref{sec:Tools}. The final output does not depend on the order in which augmentations of edges are performed.

\begin{lemma}[Chv\'atal and Sbihi~\cite{chvatal1988recognizing}, Maffray and Reed~\cite{MaffrayReed}]\label{piece_claw-free-perfect}
A claw-free perfect graph $G$ without a clique-cutset has a stability number at most $9$ or it is an augmentation of the line-graph of a bipartite multigraph. Moreover, this property is shared by any induced subgraph of $G$.
\end{lemma}

The second part of the statement is a consequence of the fact, non-immediate from the definition given here, that the class of augmentations of line-graphs of bipartite multigraphs is closed for taking induced subgraphs.

A preliminary result on line-graphs of bipartite multigraphs needed in our proof is the following.

\begin{lemma}\label{lem:operationsLGBipMulti}
The class of line-graphs of bipartite multigraphs is closed for the following two operations: deletion of a flat edge; addition of a new vertex adjacent to a maximal clique. 
\end{lemma}

\begin{proof}
Let $L$ be the line-graph of a bipartite multigraph $B$. Without loss of generality, we can assume that $L$ is connected and has at least three vertices.
Let $xy$ be a flat edge in $L$ and let $L'$ be the graph obtained from $L$ by deleting the flat edge $xy$. Because $xy$ is flat and because of the assumption on $L$, there exists three distinct vertices $b_1, b_2, b_3$ in $B$ with $x=b_1b_2$ and $y=b_2b_3$ and such that  $b_2$ has no other neighbor than $b_1$ and $b_3$ in $B$. Consider $B'$ obtained from $B$ by replacing vertex $b_2$ by two copies, one adjacent to $b_1$ through edge $x$, the other one adjacent to $b_3$ through edge $y$. Then $L'$ is the line-graph of the bipartite multigraph $B'$.

Consider now the second operation. Let $L''$ be the graph obtained from $L$ by adding a new vertex $s$ to a maximal clique. A maximal clique of $L$ corresponds in $B$ to a set of edges all adjacent to a given vertex $b_1$ since $B$ is bipartite. Consider $B''$ obtained from $B$ by adding a new vertex $b_2$ adjacent to $b_1$ by the new edge $s$. Then $L''$ is the line-graph of the bipartite multigraph $B''$.
\end{proof}

\begin{proposition}\label{prop:elem}
Let $G$ be an augmentation of the line-graph of a bipartite multigraph. A kernel of a clique-acyclic orientation of $G$ can be computed in polynomial time.
\end{proposition}

\begin{proof}
Let $G=(V,E)$ be an augmentation of the line-graph $L$ of a bipartite multigraph and $D$ a clique-acyclic orientation of $G$.
Let $x_1y_1,\ldots,x_hy_h$ denote the  matching of flat edges of $L$ that have been augmented and let $(X_i, Y_i;E_{X_iY_i})$ be the associated cobipartite graphs used in augmentations. For every $X_i,Y_i$, consider the vertices $s_{X_i}, s_{Y_i}$ and the set $S_{U_i}$ introduced in Section~\ref{sec:Tools}. For $i \in \{0,\ldots,h\}$ define the set
$$
Z_i = \left(V \backslash \bigcup_{j=1}^i (X_j \cup Y_j)\right) \cup \bigcup_{j=1}^i (\{ \sXj, \sYj\} \cup \SUj).
$$
According to Lemma~\ref{lem:augmentation}, a kernel of $D[Z_i]$ is a kernel of $D[Z_{i-1}]$ for any $i\in\{1,\ldots,h\}$. Hence a kernel of $D[Z_h]$ is a kernel of $D[Z_0] = D$. We are going to prove that $G[Z_h]$ is the line-graph of a bipartite multigraph. It will be then clear that a kernel of $D$ can be computed in polynomial time: Maffray and Reed~\cite{MaffrayReed} showed how to retrieve $L$ and the augmentations from $G$ in polynomial time, from which we can compute the set $Z_h$; the Gale-Shapley algorithm can then be used to compute a kernel of $D[Z_h]$ since $D[Z_h]$ is a clique-acyclic orientation of the line-graph of a bipartite multigraph; this kernel is then a kernel of $D$.

Consider a graph $L^*$ defined from the line-graph $L$ by applying the following operations for every $i \in \{1, \dots, h\}$: if $s_{X_i}$ and $s_{Y_i}$ are not neighbors in $D$, then delete the edge $x_i y_i$; if $s_{U_i}$ exists and it is distinct from $s_{Y_i}$, then add an additional vertex with neighborhood $N_L(y_i) \cup \{y_i\} \setminus \{x_i\}$. By repeated application of Lemma~\ref{lem:augmentation}, it comes that $G[Z_h]$ is isomorphic to the graph $L^*$. Also by repeated application of Lemma~\ref{lem:operationsLGBipMulti}, since $N_L(y_i) \cup \{y_i\} \setminus \{x_i\}$ induces a maximal clique of $L$ for every $i \in \{1, \dots, h\}$, the graph $L^*$ is the line-graph of a bipartite multigraph. Hence so is $G[Z_h]$.
\end{proof}

\begin{proof}[Proof of Theorem~\ref{clawFree}]
Let $\Drond$ be the class of clique-acyclic orientations of claw-free perfect graphs, and let $\Drond_0$ be the class of induced subdigraphs of digraphs in $\Drond$ without a clique-cutset.  According to Lemma~\ref{piece_claw-free-perfect}, any digraph $D$ in $\Drond_0$ has a stability number at most $9$ or it is an augmentation of the line-graph of a bipartite multigraph. Computing a kernel of $D$ in polynomial time can then be done as follows: test all possible subsets of at most $9$ vertices (brute-force enumeration); if it is not successful, then it means that $D$ is an augmentation of the line-graph of a bipartite graph and computing a kernel can be done in polynomial time according to Proposition~\ref{prop:elem}. Proposition~\ref{prop:atomes} implies thus that computing a kernel can be done over the whole class $\Drond$ in polynomial time.
\end{proof}

\section{Chordal, circular-arc, and other intersection graphs}
\label{sec:inter}


 \subsection{Perfect clique-acyclic case}
 \label{subsec:interParfait}

\begin{proof}[Proof of Theorem~\ref{ChordalCASuper}]
Chordal graphs without clique-cutset are cliques and any sink of a clique-acyclic super-orientation of a clique is a kernel of that clique. Proposition~\ref{prop:atomes} is applicable with $\Drond$ defined as the class of all clique-acyclic super-orientations of chordal graphs and $\Drond_0$ defined as the class of clique-acyclic super-orientations of cliques.
\end{proof}

Another class where a similar result holds is the class of \emph{Directed Edge graphs}, or \emph{DE graphs}. DE graphs are intersection graphs of directed paths (seen as arc sets) in a directed tree. They have been introduced and studied by Monma and Wei \cite{MonmaWei}. They proved that DE graphs are perfect and thus subject to the Boros-Gurvich theorem. For orientations, polynomiality of kernel computation has been established by Durand de Gevigney et al. \cite{Meunier}. We settle polynomiality in the general case of super-orientations.

\begin{theorem}\label{thm:DE}
Kernel computation is polynomial in clique-acyclic super-orien\-ta\-tions of DE graphs.
\end{theorem}
 
\begin{proof} 
DE graphs without clique-cutset are line-graphs of bipartite multigraphs and kernels of clique-acyclic super-orientations of such graphs are computable in polynomial time with the Gale-Shapley algorithm. Proposition~\ref{prop:atomes} is applicable with $\Drond$ defined as the class of all clique-acyclic super-orientations of DE graphs and $\Drond_0$ defined as the class of clique-acyclic super-orientations of line-graphs of bipartite multigraphs.
\end{proof}


 \subsection{Non-clique-acyclic orientations}
\label{subsec:interNonCA}

We present additional results for non-necessarily clique-acyclic orientations of chordal and claw-free graphs.

\begin{proposition}\label{uniciteChordal}
Any orientation of a chordal graph has at most one kernel.
\end{proposition}

\begin{proof}
Let $D$ be an orientation of a chordal graph $G$.
Assume for a contradiction that $D$ has two different kernels $K$ and $K'$. Consider the induced subgraph $G[K \Delta K']$: it is chordal since $G$ is chordal. It is also bipartite because $K\backslash K'$ and $K' \backslash K$ form a partition of the vertices into two stable sets. A graph that is both chordal and bipartite is cycle-free. Hence its orientation $D[K \Delta K']$ is a digraph without directed cycle. This digraph is not empty hence it contains a sink $s$, and w.l.o.g. $s \in K \backslash K'$. The vertex $s$ is not absorbed by any vertex in $K' \backslash K$, and by stability of $K$ it is not absorbed by a vertex in $K' \cup K$. This contradicts the fact that $K'$ is a kernel of $D$. 
\end{proof}

We noted that a similar proof leads to a new result on cardinality of kernels in claw-free graphs.

\begin{proposition}
\label{prop:cardClawFree}
All kernels of an orientation of a claw-free graph have the same size.
\end{proposition}

\begin{proof}
Let $D$ be an orientation of a claw-free graph $G$.
Consider two different kernels $K$ and $K'$ of $D$. The induced subgraph $G[K \Delta K']$ is claw-free, and it is bipartite. Hence each vertex in this graph has degree at most two. It follows that $G[K \Delta K']$ is a disjoint union of cycles and paths. Assume for a contradiction that one of the connected components of $G[K \Delta K']$ is a path $P$. This path has a sink $s$, and we can assume w.l.o.g. $s \in K \backslash K'$. Then $s$ is not absorbed by any other vertex in $K' \setminus K$. It is neither absorbed by a vertex in $K' \cup K$ by stability of $K$. This contradicts the fact that $K'$ is a kernel of $D$, hence such a path $P$ does not exist. Then $G[K \Delta K']$  is a disjoint union of cycles, which alternate between $K$ and $K'$, hence $|K| = |K'|$. 
\end{proof}

Note that this result generalizes a well-known property of stable matchings: all stable matchings have the same size, as noted by Schrijver \cite[Corollary 18.12a]{Schrijver}, who credits this remark to McVitie and Wilson~\cite{McVitieWilson}.

An orientation of a chordal graph does not necessarily have a kernel when the orientation is not clique-acyclic. Though, with the help of Proposition~\ref{uniciteChordal}, we prove that deciding the existence of a kernel and computing it if it exists is polynomial.

\begin{proposition}
\label{algoChordal}
Let $D=(V,A)$ be an orientation of a chordal graph. Deciding whether $D$ has a kernel and computing such a kernel if it exists can be performed in polynomial time.
\end{proposition}

\begin{proof}
The algorithm is the following.

If $D$ is a clique, then search for a sink: return it or answer that none exists.

If $D$ is not a clique, then proceed as follows.
Compute a \emph{simplicial vertex} $v$, i.e., a vertex whose neighborhood induces a clique. Such a vertex can be found in polynomial time. 
Let  $U$ be the set of the non-neighbors of $v$. Call recursively the algorithm on $D[U \cup N_D^+(v)]$.
\begin{enumerate}[label=(\alph*)]
\item\label{a} If $D[U \cup N_D^+(v)]$ has no kernel, then answer that $D$ has no kernel.
\item\label{b} Otherwise, the recursive call has returned a kernel $K'$ of $D[U \cup N_D^+(v)]$.
\begin{enumerate}[label=(\roman*)]
\item\label{i} If $K'$ is a kernel of $D$, then return $K'$.
\item If $K'\cup\{v\}$ is a kernel of $D$, then return $K'\cup\{v\}$.
\item Otherwise, answer that $D$ has no kernel.
\end{enumerate}
\end{enumerate}

This algorithm is obviously polynomial. Let us prove that it is correct.

We show that if $D$ has a kernel $K$, then $K\setminus\{v\}$ is a kernel of $D[U\cup N_D^+(v)]$. Consider a kernel $K$ of $D$. If $v\in K$, then $K\setminus\{v\}\subseteq U\cup N_D^+(v)$ and the vertices in $U\cup N_D^+(v)$ are not absorbed by $v$. If $v\notin K$, then $K\cap N_D^-(v)=\varnothing$ since $v$ has to be absorbed, and $K\setminus\{v\}=K\subseteq U\cup N_D^+(v)$. In both cases $K\setminus\{v\}$ is a kernel of $D[U\cup N_D^+(v)]$.
Thus in case~\ref{a}, the algorithm outputs the correct answer.

Moreover, note that it also implies that if $D$ has a kernel $K$, then $K'=K\setminus\{v\}$ by uniqueness of the kernel of $D[U \cup N^+_D(v)]$ (Proposition~\ref{uniciteChordal}). Hence $K'$ and $K' \cup \{v\}$ are the only possible kernels of $D$, which leads to the conclusion.
\end{proof}

A graph is a \emph{circular-arc graph} if it is the intersection graph of intervals on a circle. Circular-arc graphs are not necessarily perfect.

\begin{proposition}
Let $D=(V,A)$ be an orientation of a circular-arc graph. Deciding whether $D$ has a kernel and computing such a kernel if it exists can be performed in polynomial time.
\end{proposition}

\begin{proof}
The algorithm is the following. Fix a point on the circle in the representation of $D$ and let $C$ be the clique of all intervals crossing this point. For every set $S \subseteq C$ with $|S| \leq 1$, consider $D_S$ the subdigraph induced by vertices not in $C $ and not neighbors of $S$. Search for a kernel $K_S$ of $D_S$: the digraph $D_S$ is an orientation of an interval graph, thus the algorithm of Proposition~\ref{algoChordal} is applicable. If $K_S$ exists and $S \cup K_S$ absorbs all vertices in $D$, return $S \cup K_S$.
After testing all sets $S$, if nothing has been returned so far, then answer that $D$ has no kernel.

For every $S$, the digraph $D_S$ is chordal, hence by Proposition~\ref{uniciteChordal} the kernel $K_S$ is its unique kernel. Any kernel $K$ of $D$ can be decomposed as $K = S \cup K_S$, with $S \subseteq C$, $|S| \leq 1$ and $K_S$ the kernel of the digraph $D_S$. This ensures that the algorithm always returns the kernel of $D$ if it exists.
\end{proof}


\bibliographystyle{plain}
\bibliography{kernel}

\begin{thebibliography}{10}

\bibitem{Abbas}
M.~Abbas and Y.~Saoula.
\newblock Polynomial algorithms for kernels in comparability, permutation and
  $\text{P}_4$-free graphs.
\newblock {\em 4OR}, 3:217--225, 2005.

\bibitem{AhHo}
R.~Aharoni and R.~Holzman.
\newblock Fractional kernels in digraphs.
\newblock {\em Journal of Combinatorial Theory, Series B}, 73:1--6, 1998.

\bibitem{AnHo14}
S.~D. Andres and W.~Hochst\"attler.
\newblock Perfect digraphs.
\newblock {\em Journal of Graph Theory}, 79:21--29, 2014.

\bibitem{BorosGurvich}
E.~Boros and V.~Gurvich.
\newblock Perfect graphs are kernel solvable.
\newblock {\em Discrete Mathematics}, 159:35--55, 1996.

\bibitem{BorosGurvichAppli}
E.~Boros and V.~Gurvich.
\newblock Perfects graphs, kernels, and cores of cooperative games.
\newblock {\em Discrete Mathematics}, 306:2336--2354, 2006.

\bibitem{ChvatalNPC}
V.~Chv\'atal.
\newblock On the computational complexity of finding a kernel.
\newblock Technical Report CRM300, Centre de Recherches Math\'ematiques,
  Universit\'e de Montr\'eal, 1973.

\bibitem{chvatal1988recognizing}
V.~Chv{\'a}tal and N.~Sbihi.
\newblock Recognizing claw-free perfect graphs.
\newblock {\em Journal of Combinatorial Theory, Series B}, 44(2):154--176,
  1988.

\bibitem{Meunier}
O.~Durand~de Gevigney, F.~Meunier, C.~Popa, J.~Reygner, and A.~Romero.
\newblock Solving coloring, minimum clique cover and kernel problems on arc
  intersection graphs of directed paths on a tree.
\newblock {\em 4OR}, 9:175--88, 2011.

\bibitem{Egres}
{Egerv\'ary Research Group}.
\newblock {E}gres {O}pen, finding kernels in special digraphs.
\newblock
  \url{http://lemon.cs.elte.hu/egres/open/Finding_kernels_in_special_digraphs},
  2016.

\bibitem{Fraenkel}
A.~S. Fraenkel.
\newblock Planar kernel and grundy with $d \leq 3$, $d^{out} \leq 2$, $d^{in}
  \leq 2$ are {N}{P}-complete.
\newblock {\em Discrete Applied Mathematics}, 3(4):257 -- 262, 1981.

\bibitem{GaleShapley}
D.~Gale and L.~S. Shapley.
\newblock College {A}dmissions and the {S}tability of {M}arriage.
\newblock {\em The American Mathematical Monthly}, 69:9--15, 1962.

\bibitem{Galvin}
F.~Galvin.
\newblock The list chromatic index of a bipartite multigraph.
\newblock {\em Journal of Combinatorial Theory, Series B}, 63:153--158, 1995.

\bibitem{Jacob}
H.~Jacob.
\newblock Kernels in graphs with a clique-cutset.
\newblock {\em Discrete Mathematics}, 156:265--267, 1996.

\bibitem{kintali2009reducibility}
S.~Kintali, L.~J. Poplawski, R.~Rajaraman, R.~Sundaram, and S.-H. Teng.
\newblock Reducibility among fractional stability problems.
\newblock In {\em 50th Annual IEEE Symposium on Foundations of Computer
  Science}, pages 283--292, 2009.

\bibitem{KiralyPap}
T.~Kir\'aly and J.~Pap.
\newblock A note on kernels and {S}perner's {L}emma.
\newblock {\em Discrete Applied Mathematics}, 157:3327--3331, 2009.

\bibitem{MaffrayITriang}
F.~Maffray.
\newblock On kernels in $i$-triangulated graphs.
\newblock {\em Discrete Mathematics}, 61:247--251, 1986.

\bibitem{maffray1992kernels}
F.~Maffray.
\newblock Kernels in perfect line-graphs.
\newblock {\em Journal of Combinatorial Theory, Series B}, 55(1):1--8, 1992.

\bibitem{MaffrayReed}
F.~Maffray and B.~A. Reed.
\newblock A description of claw-free perfect graphs.
\newblock {\em Journal of Combinatorial Theory, Series B}, 75(1):134 -- 156,
  1999.

\bibitem{McVitieWilson}
D.~G. McVitie and L.~B. Wilson.
\newblock Stable marriage assignment for unequal sets.
\newblock {\em BIT Numerical Mathematics}, 10(3):295--309, 1970.

\bibitem{MonmaWei}
C.~L. Monma and V.~K. Wei.
\newblock Intersection graphs of paths in a tree.
\newblock {\em Journal of Combinatorial Theory, Series B}, 41:141--181, 1986.

\bibitem{VonNeumannMorgenstern}
O.~Morgenstern and J.~von Neumann.
\newblock {\em Theory of Games and Economic Behavior}.
\newblock Princeton University Press, 1944.

\bibitem{Richardson}
M.~Richardson.
\newblock On weakly ordered systems.
\newblock {\em Bulletin of the American Mathematical Society}, 52(2):113--116,
  1946.

\bibitem{Schrijver}
A.~Schrijver.
\newblock {\em Combinatorial {O}ptimization}.
\newblock Springer, 2003.

\bibitem{Tarjan}
R.~E. Tarjan.
\newblock Decomposition by clique separators.
\newblock {\em Discrete Mathematics}, 55:221--232, 1985.

\end{thebibliography}

\end{document}